\documentclass[english,11pt]{article}

\usepackage[english]{babel}
\usepackage[utf8x]{inputenc}
\usepackage[T1]{fontenc}
\usepackage[formats]{listings}
\usepackage{geometry}
\usepackage{enumitem}

\usepackage{graphicx}
\usepackage[colorinlistoftodos]{todonotes}
\usepackage{setspace}
\usepackage{placeins}
\usepackage{enumitem}
\usepackage{titlesec}
\usepackage{cite}
\usepackage{comment}
\usepackage{caption}
\usepackage{subcaption}

\usepackage[colorlinks=true, allcolors=blue]{hyperref}

\usepackage{amsmath}
\usepackage{amsthm}
\usepackage{amssymb}
\usepackage{amsfonts}
\usepackage{bbm}
\usepackage{bm}
\usepackage{nicefrac}
\usepackage{mathtools}

\usepackage{tikz}

\usepackage{graphicx}
\usepackage{fullpage}
\usepackage{float}
\usepackage{wrapfig}
\usepackage{caption}

\usepackage{algpseudocode}
\usepackage{algorithm}
\usepackage{listings}

\usepackage{amsthm}\usepackage{dsfont}\usepackage{array}\usepackage{mathrsfs}\usepackage{cite}\usepackage{comment}\usepackage{mathrsfs}

\makeatother

\usepackage{babel}
\usepackage{color}
\usepackage{centernot}

\usepackage{babel}
\usepackage{sansmath}

\definecolor{ejc}{RGB}{255,0,0}
\definecolor{ps}{RGB}{0,0,200}

\title{The Phase Transition for the Existence of the Maximum
  Likelihood Estimate in High-dimensional Logistic Regression}

\author{Emmanuel
  J. Cand\`es\thanks{Department of Statistics, Stanford
    University, Stanford, CA 94305, U.S.A.} \thanks{Department of Mathematics,
    Stanford University, Stanford, CA 94305, U.S.A.} 
  \and Pragya Sur\footnotemark[1]}

\theoremstyle{plain}\newtheorem{lemma}{\textbf{Lemma}}\newtheorem{theorem}{\textbf{Theorem}}\newtheorem{proposition}{\textbf{Proposition}}

\theoremstyle{definition}

\theoremstyle{definition}

\newcommand{\bbeta}{\bm{\beta}}

\newcommand{\bX}{\bm{X}}

\newcommand{\bzero}{\bm{0}}

\newcommand{\dnorm}{{\mathcal{N}}}

\newcommand{\iid}{\stackrel{ \text{i.i.d.} }{\sim} }

\newcommand{\bu}{\bm{u}}

\newcommand{\R}{\mathbb{R}}

\newcommand{\eqd}{\stackrel{\mathrm{d}}{=}}

\newcommand{\by}{\bm{y}}
\newcommand{\bx}{\bm{x}}

\newcommand{\bA}{\bm{A}}

\newcommand{\bb}{\bm{b}}

\newcommand{\bZ}{\bm{Z}}

\newcommand{\bw}{\bm{w}}

\newcommand{\convd}{\stackrel{\mathrm{d}}{\rightarrow}}
\newcommand{\bV}{\bm{V}}

\newcommand{\bSigma}{\bm{\Sigma}}

\newcommand{\prob}{\mathbb{P}}

\newcommand{\convP}{\stackrel{\prob}{\longrightarrow}}

\newcommand{\bY}{\bm{Y}}

\newcommand{\bone}{\bm{1}}

\def\independenT#1#2{\mathrel{\rlap{$#1#2$}\mkern2mu{#1#2}}}
\newcommand\independent{\protect\mathpalette{\protect\independenT}{\perp}}

\newcommand{\One}[1]{{\mathbbm{1}}\left\{{#1}\right\}}

\newcommand{\Var}{\operatorname{Var}}
\newcommand{\goto}{\rightarrow}
\newcommand{\lspan}{\operatorname{span}}
\newcommand{\hmle}{h_{\text{MLE}}}

\makeatletter
\renewenvironment{proof}[1][\proofname] {
	\par\pushQED{\qed}\normalfont
	\topsep6\p@\@plus6\p@\relax
	\trivlist\item[\hskip\labelsep\bfseries#1\@addpunct{:}]
 	\ignorespaces
} {
	\popQED\endtrivlist\@endpefalse
}
\makeatother

\DeclareMathOperator{\E}{\mathbb{E}}

\begin{document}
\maketitle

\begin{abstract}
  This paper rigorously establishes that the existence of the maximum
  likelihood estimate (MLE) in high-dimensional logistic regression
  models with Gaussian covariates undergoes a sharp `phase
  transition'. We introduce an explicit boundary curve $\hmle$,
  parameterized by two scalars measuring the overall magnitude of the
  unknown sequence of regression coefficients, with the following
  property: in the limit of large sample sizes $n$ and number of
  features $p$ proportioned in such a way that $p/n \goto \kappa$, we
  show that if the problem is sufficiently high dimensional in the
  sense that $\kappa > \hmle$, then the MLE does not exist with
  probability one. Conversely, if $\kappa < \hmle$, the MLE
  asymptotically exists with probability one.
\end{abstract}

\section{Introduction}
\label{sec:introduction}

Logistic regression
\cite{nelder1972generalized,mccullagh1989generalized} is perhaps the
most widely used and studied non-linear model in the multivariate
statistical literature. For decades, statistical inference for this
model has relied on likelihood theory, especially on the theory of
maximum likelihood estimation and of likelihood ratios.  Imagine we
have $n$ independent observations $(\bx_i,y_i)$, $i = 1, \ldots, n$, 
where the response $y_i \in \{-1,1\}$ is linked to the covariates
$\bx_i \in \R^p$ via the logistic model
\[
  \mathbb{P}(y_i = 1 | \bx_i) =
  \sigma(\bx_i' \bbeta), \qquad \sigma(t) :=\frac{e^t}{1 + e^{t}}; 
\]
here, $\bbeta \in \R^p$ is the unknown vector of regression
coefficients. In this model, the log-likelihood is given by
\[
\ell(\bb) = \sum_{i = 1}^n -\log(1 + \exp(-y_i \bx_i'\bb))
\]
and, by definition, the maximum likelihood estimate (MLE) is any
maximizer of this functional.

\subsection{Data geometry and the existence of the MLE}
\label{sec:geometry}

The delicacy of ML theory is that the MLE does not exist in all
situations, even when the number $p$ of covariates is much smaller
than the sample size $n$. This is a well-known phenomenon, which
sparked several interesting series of investigation. One can even say
that characterizing the existence and uniqueness of the MLE in
logistic regression has been a classical problem in statistics. For
instance, every statistician knows that if the $n$ data points
$(\bx_i, y_i)$ are {\em completely separated} in the sense that that
there is a linear decision boundary parameterized by $\bb \in \R^p$
with the property
\begin{equation}
  \label{eq:separated}
  y_i \bx_i' \bb > 0, \text{ for all } i, 
\end{equation}
then the MLE does not exist. To be clear, \eqref{eq:separated} means
that the decision rule that assigns a class label equal to the sign of
$\bx_i'\bb$ makes no mistake on the sample. Every statistician also
knows that if the data points {\em overlap} in the sense that for
every $\bb \neq \bzero$, there is at least one data point that is
classified correcly ($y_i \bx_i' \bb > 0$) and at least another that
is classified incorrectly ($y_k \bx_k' \bb < 0$), then the MLE does
exist. The remaining situation, where the data points are {\em
  quasi-completely separated}, is perhaps less well-known to
statisticians: this occurs when for any decision $\bb \neq
\bzero$, \begin{equation}
  \label{eq:quasi}
  y_i \bx_i' \bb \ge 0,  \text{ for all } i, 
\end{equation}
where equality above holds for some of the observations. A useful
theorem of Albert and Anderson \cite{Albert1984MLE} states that the
MLE does not exist in this case either. {\em Hence, the MLE exists if
  and only if the data points overlap.}

Historically, \cite{Albert1984MLE} follows earlier work of Silvapulle
\cite{silvapulle1981existence}, who proposed necessary and sufficient
conditions for the existence of the MLE based on a geometric
characterization involving convex cones (see \cite{Albert1984MLE} for
additional references).  Subsequently, Santner and Duffy
\cite{santner1986note} expanded on the characterization from
\cite{Albert1984MLE} whereas Kaufman \cite{kaufmann1988existence}
established theorems on the existence and uniqueness of the minimizer
of a closed proper convex function. 
In order to detect separation, linear programming approaches have been
proposed on multiple occasions, see for instance,
\cite{Albert1984MLE,silvapulle1986existence,konis2007linear}. Detection
of complete separation was studied in further detail in
\cite{lesaffre1989partial,kolassa1997infinite}. Finally,
\cite{christmann2001measuring} analyzes the notion of regression depth
for measuring overlap in data sets.

\subsection{Limitations}

Although beautiful, the aforementioned geometric characterization does
not concretely tell us when we can expect the MLE to exist and when we
cannot. Instead, it trades one abstract notion, ``there is an MLE'',
for another, ``there is no separating hyperplane''.  To drive our
point home, imagine that we have a large number of covariates $\bx_i$,
which are independent samples from some distribution $F$, as is almost
always encountered in modern applications.  Then by looking at the
distribution $F$, the data analyst would like to be able to predict
when she can expect to find the MLE and she cannot. The problem is
that the abstract geometric separation condition does not inform her
in any way; she would have no way to know a priori whether the MLE
would go to infinity or not.

\subsection{Cover's result}

One notable exception against this background dates back to the
seminal work of Cover \cite{cover1964thesis,cover1965geometrical}
concerning the separating capacities of decision surfaces. When
applied to logistic regression, Cover's main result states the
following: assume that the $\bx_i$'s are drawn i.i.d.~from a
distribution $F$ obeying some specific assumptions and that the {\em
  class labels are independent from $\bx_i$} and have equal marginal
probabilities; i.e.~$\mathbb{P}(y_i = 1 | \bx_i) = 1/2$.  Then Cover
  shows that as $p$ and $n$ grow large in such a way that
  $p/n \rightarrow \kappa$, the data points asymptotically
  overlap---with probability tending to one---if $\kappa < 1/2$
  whereas they are separated---also with probability tending to
  one---if $\kappa > 1/2$. In the former case where the MLE exists,
  \cite{sur2017likelihood} refined Cover's result by calculating the
  limiting distribution of the MLE when the features $\bx_i$ are
  Gaussian.

Hence, the results from \cite{cover1964thesis,cover1965geometrical}
and \cite{sur2017likelihood} describe a phase transition in the
existence of the MLE as the dimensionality parameter $\kappa = p/n$
varies around the value 1/2. Therefore, a natural question is this:
\begin{center}
{\em Do phase
transitions exist in the case where the class labels $y_i$ actually
\underline{depend} on the features $\bx_i$?}
\end{center} Since likelihood based inference
procedures are used all the time, it is of significance to understand
when the MLE actually exists. This paper is about this question.

\subsection{Phase transitions}
\label{sec:pt} 


This work rigorously establishes the existence of a phase transition
in the logistic model with Gaussian covariates, and computes the phase
transition boundary explicitly.

\paragraph{Model} Since researchers routinely include an intercept in
the fitted model, we consider such a scenario as well. Throughout the
paper, we assume we have $n$ samples $(\bx_i, y_i)$ with Gaussian
covariates: 
\[
\bx_i \iid
\mathcal{N}(\bzero, \bSigma), \quad \mathbb{P}(y_i = 1 | \bx_i) =
\sigma(\beta_0 + \bx_i' \bbeta) = 1 - \mathbb{P}(y_i = -1 | \bx_i), 
\]
where the covariance $\bSigma$ is non-singular but otherwise
arbitrary.

\paragraph{Peek at the result} To describe our results succinctly,
assume the high-dimensional asymptotics from the previous section in
which $p/n \rightarrow \kappa$ (assumed to be less than one throughout
the paper). To get a meaningful result in diverging dimensions, we
consider a sequence of problems with $\beta_0$ fixed and
\begin{equation}
  \label{eq:snr1}
  \operatorname{Var}(\bx_i' \bbeta) \rightarrow \gamma_0^2.
\end{equation}
This is set so that the log-odds ratio $\beta_0 + \bx_i' \bbeta$ does
not increase with $n$ or $p$, so that the likelihood is not trivially
equal to either $0$ or $1$.  Instead,
\begin{equation}
  \label{eq:snr2}
  \sqrt{\E (\beta_0 + \bx_i' \bbeta)^2} \rightarrow \sqrt{\beta_0^2 +
    \gamma_0^2}  =: \gamma.
\end{equation}
In other words, we put ourselves in a regime where accurate estimates
of $\bbeta$ translate into a precise evaluation of a non-trivial
probability.

Our main result is that there is an explicit function $h_{\text{MLE}}$
  given in \eqref{eq:gmle} such that
 \[
\begin{array}{ccc}
  \kappa > h_{\text{MLE}}(\beta_0, \gamma_0) 
  &  \implies \quad  &
                       \mathbb{P}\{\text{MLE
                       exists}\}
                       \rightarrow
                       0, \\
  \kappa < h_{\text{MLE}}(\beta_0,\gamma_0)  
  & \implies \quad  &
    \mathbb{P}\{\text{MLE
    exists}\}
    \rightarrow
    1.
\end{array}                     
\]
Hence, the existence of the MLE undergoes a sharp change: below the
curves shown in Figure \ref{fig:formula}, the existence probability
asymptotically approaches $1$; above, it approaches $0$. Also note
that the phase-transition curve depends upon the unknown regression
sequence $\bbeta \in \R^p$ only through the intercept $\beta_0$ and
$\gamma_0^2 = \lim_{n, p \rightarrow \infty}
\operatorname{Var}(\bx_i'\bbeta)$.


The formula for the phase transition $h_{\text{MLE}}$ is new. As we
will see, it is derived from ideas from convex geometry.



\section{Main Result}
\label{sec:formula}

\subsection{Model with intercept}

Throughout the paper, for each $\beta_0 \in \R$ and $\gamma_0 \ge 0$,
we write
\begin{equation}
\label{eq:V}
(Y, V) \sim F_{\beta_0, \gamma_0} \quad \text{if} \quad (Y, V)  \eqd (Y, YX), 
\end{equation}
where $X \sim \dnorm(0,1)$, and
$\mathbb{P}(Y = 1 | X) = 1 - \mathbb{P}(Y = -1 | X) = \sigma(\beta_0 +
\gamma_0 X)$. 

\begin{figure}
\begin{center}
	\begin{tabular}{ccc}
          \includegraphics[width=6.5cm,keepaspectratio]{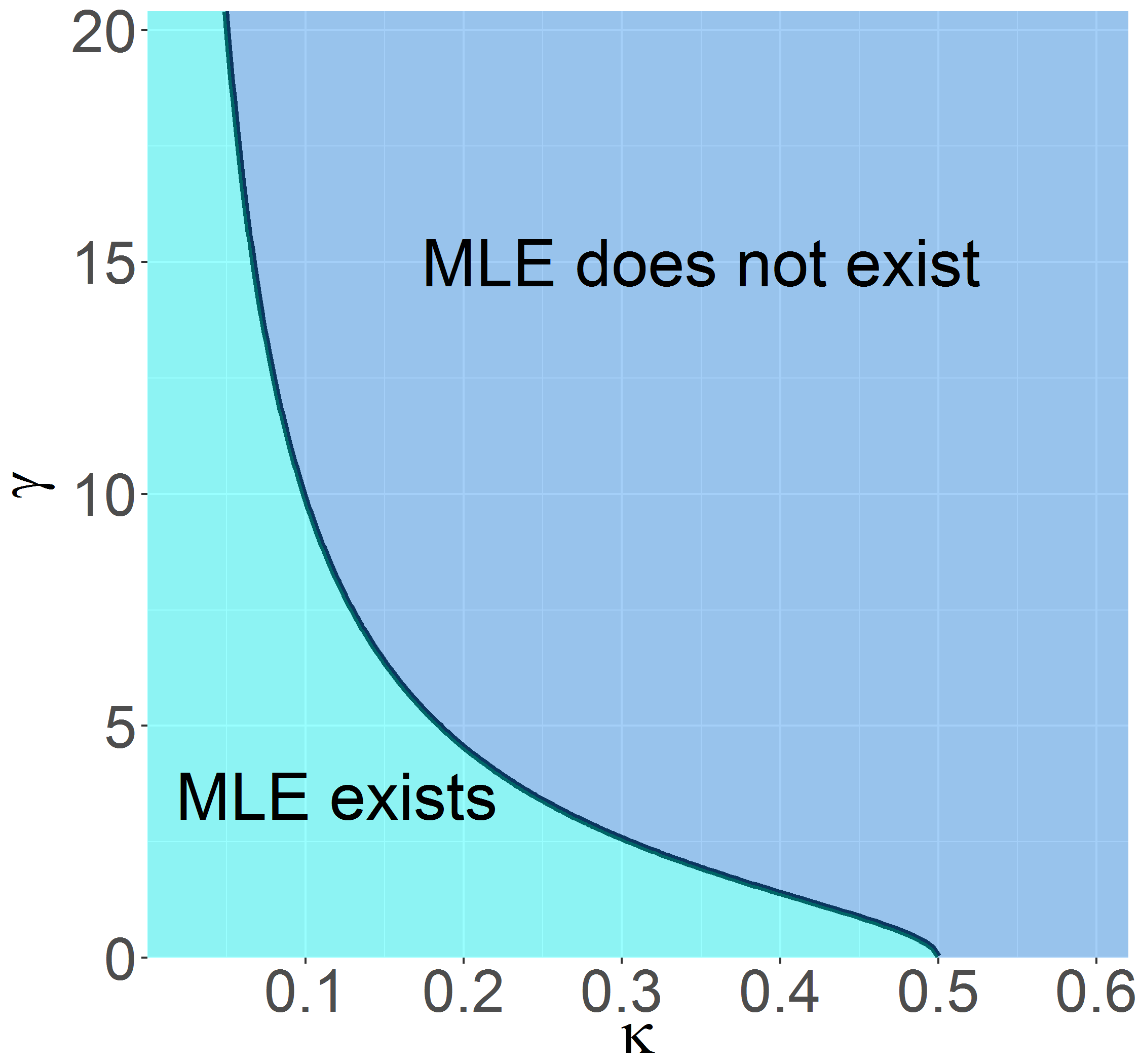} 
          & & \includegraphics[width=6.5cm,keepaspectratio]{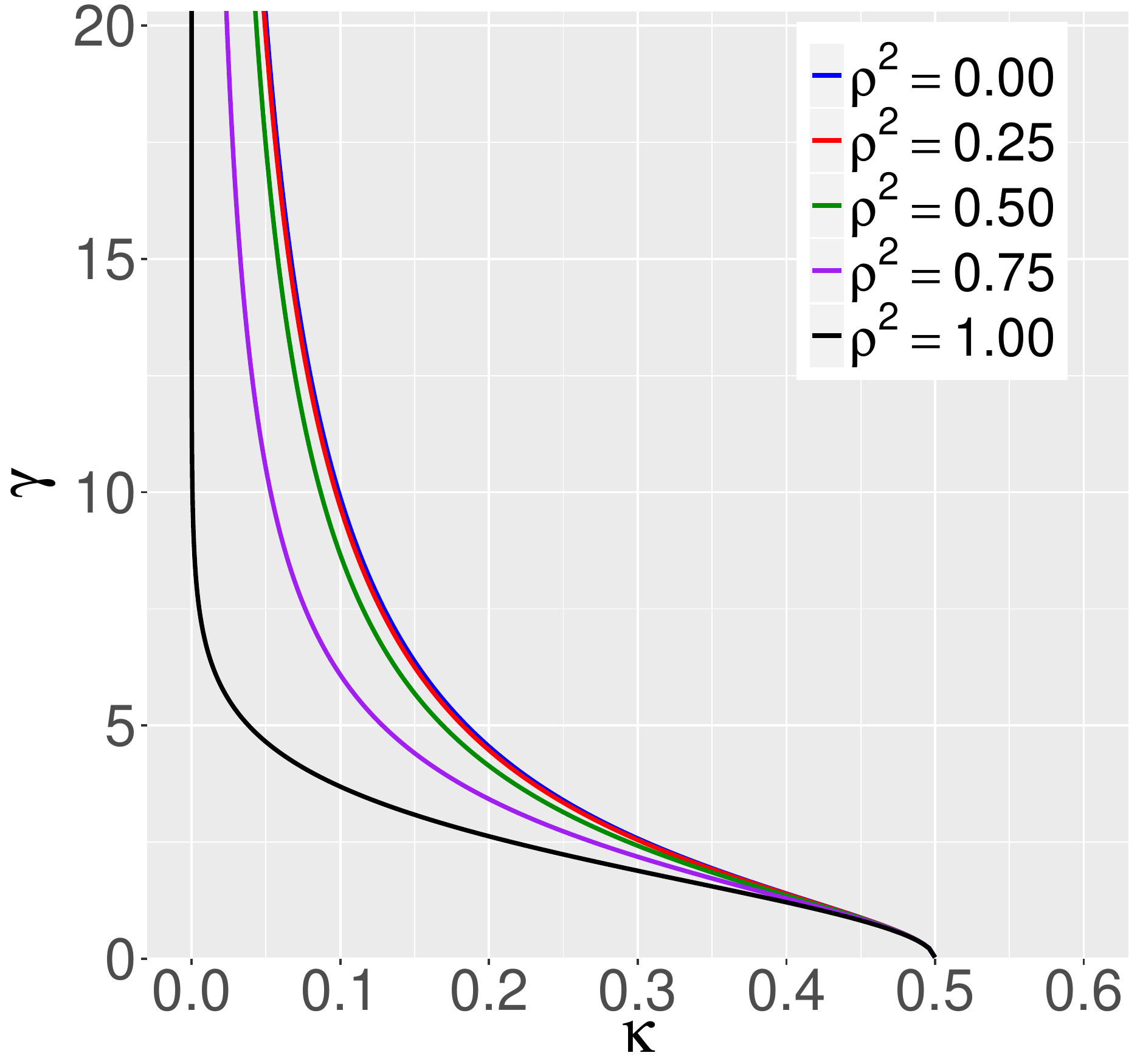} \tabularnewline
              (a) & & (b)   \tabularnewline
	\end{tabular}
\end{center}
\caption{Theoretical predictions from \eqref{eq:gmle}. (a) Boundary
  curve $\gamma \mapsto \hmle(0,\gamma)$ separating the regions where
  the MLE asymptotically exists and where it does not (in this case
  $\beta_0 = 0$). (b) Boundary curves
  $\gamma \mapsto \hmle(\rho \gamma,\sqrt{1-\rho^2} \gamma)$ for
  various values of $\rho$. The curve with $\rho = 0$ shown in blue is
  that from (a). It is hardly visible because it is close to that with
  $\rho^2 = 0.25$.}
\label{fig:formula}
\end{figure}

 \begin{theorem}
   \label{thm:main} Let $(Y,V) \sim F_{\beta_0,\gamma_0}$ and
   $Z \sim \dnorm(0,1)$ be independent random variables. Define
   \begin{equation}
     \label{eq:gmle}
     h_{\text{\em MLE}}(\beta_0,\gamma_0) = \min_{t_0, t_1 \in \R} \, \left\{\E (t_0 Y + t_1 V -
       Z)_+^2\right\}, 
   \end{equation}
   where $x_+ = \operatorname{max}(x,0)$ and we write
   $x_+^2 = (x_+)^2$ for short.
   Then in the setting from Section \ref{sec:pt}, 
  \[
\begin{array}{lll}
  \kappa > h_{\text{\em MLE}}(\beta_0, \gamma_0) 
  &  \implies \quad 
  &
    \lim_{n,p
    \rightarrow
    \infty} 
    \mathbb{P}\{\text{\em MLE
    exists}\}
    =
    0, \\
  \kappa < h_{\text{\em MLE}}(\beta_0,\gamma_0)  
  &  \implies \quad 
  &
    \lim_{n,p
    \rightarrow
    \infty} 
    \mathbb{P}\{\text{\em MLE
    exists}\}
    =
    1.
\end{array}                     
\]
\end{theorem}
This result is proved in Section \ref{sec:conic_geom}. As the reader
will gather from checking our proof, our convergence result is
actually more precise. We prove that the transition occurs in an
interval of width $O(n^{-1/2})$: take any sequence
$\lambda_n \goto \infty$; then
 \[
\begin{array}{lll}
  p/n > h_{\text{MLE}}(\beta_0, \gamma_0) + \lambda_n n^{-1/2} 
  &  \implies \quad 
  &
    \lim_{n,p
    \rightarrow
    \infty} 
    \mathbb{P}\{\text{MLE
    exists}\}
    =
    0, \\
  p/n < h_{\text{MLE}}(\beta_0, \gamma_0) - \lambda_n n^{-1/2}   
  &  \implies \quad 
  &
    \lim_{n,p
    \rightarrow
    \infty} 
    \mathbb{P}\{\text{MLE
    exists}\}
    =
    1.
\end{array}
\]  

It is not hard to see that $h_{\text{MLE}}$ defined for values of
$\beta_0 \in \R$ and $\gamma_0 \ge 0$ is symmetric in its first
argument,
$h_{\text{MLE}}(\beta_0, \gamma_0) = h_{\text{MLE}}(-\beta_0,
\gamma_0)$.
We thus only consider the case where $\beta_0 \ge 0$. Over the
non-negative orthant $\R^2_+$, $h_{\text{MLE}}(\beta_0, \gamma_0)$ is
a decreasing function of both $\beta_0$ and $\gamma_0$. Figure
\ref{fig:formula} shows a few phase-transition curves.

\subsection{Special cases}

It is interesting to check the predictions of formula \eqref{eq:gmle}
for extreme values of $\gamma := \sqrt{\beta_0^2 + \gamma_0^2}$,
namely, $\gamma= 0$ (no signal) and $\gamma \goto \infty$ (infinite
signal).
\begin{itemize}
\item At $\gamma = 0$, $Y$ and $V$ are independent, and $Y$ is a
  Rademacher variable whereas $V$ is a standard Gaussian. The variable
  $t_0 Y + t_1 V - Z$ is, therefore, symmetric and
  \[
    h_{\text{MLE}} (0,0) = \min_{t_0,t_1} \, \frac{1}{2} \E (t_0 Y +
      t_1 V - Z)^2= \min_{t_0,t_1} \, \frac{1}{2} (t_0^2 + t_1^2 + 1) 
      = \frac{1}{2}.
  \]
  Hence, this recovers and extends Cover's result: in the limit where
  $\beta_0^2 + \bbeta' \bSigma \bbeta \goto 0$ (this includes the case
  where $y_i$ is symmetric and independent of $\bx_i$ as in
  \cite{cover1965geometrical, cover1964thesis}), we obtain that the
  phase transition is at $\kappa = 1/2$.
  
\item When $\gamma_0 \goto \infty$,
  $V \convd |Z'|$,
  $Z' \sim \dnorm(0,1)$. Hence, plugging $t_0 = 0$ into \eqref{eq:gmle}
  gives
  \[
    \lim_{t_1 \goto -\infty} \, \E (t_1|Z'|- Z)_+^2 = 0. 
  \]
  If $\beta_0 \goto \infty$, $Y \convd 1$ and plugging $t_1 = 0$ into
  \eqref{eq:gmle} gives
  \[
  \lim_{t_0 \goto -\infty} \, \E (t_0- Z)_+^2 = 0. 
  \]
  Either way, this says that in the limit of infinite signal strength,
  we must have $p/n \goto 0$ if we want to guarantee the existence of
  the MLE.
\end{itemize}

We simplify \eqref{eq:gmle} in other special cases below. 
\begin{lemma}
  \label{lem:special}
  In the setting of Theorem \ref{thm:main}, consider the special case
  $\gamma_0 = 0$, where the response does not asymptotically depend on
  the covariates: we have 
  \begin{equation}
    \label{eq:beta0}
     h_{\text{\em MLE}}(\beta_0,0) = \min_{t \in \R} \, \left\{\E
      (t Y - Z)_+^2\right\}.
  \end{equation}
  In the case $\beta_0 = 0$ where the marginal probabilities are
  balanced, $\mathbb{P}(y_i = 1) = \mathbb{P}(y_i = -1) = 1/2$,
  \begin{equation}
    \label{eq:gamma0}
    h_{\text{\em MLE}}(0,\gamma_0) = \min_{t \in \R} \, \left\{\E
      (t V - Z)_+^2\right\}.
  \end{equation}
\end{lemma}
\begin{proof}
  Consider the first assertion. In this case, it follows from the
  definition \eqref{eq:V} that $(Y,V) \eqd (Y,X)$ where $Y$ and $X$
  are independent, $\mathbb{P}(Y = 1) = \sigma(\beta_0)$ and
  $X \sim \dnorm(0,1)$. Hence,
\begin{align*}
  h_{\text{MLE}}(\beta_0,0) = \min_{t_0,t_1} \, \E (t_0 Y -
  \sqrt{1+t_1^2} Z)_+^2 & = \min_{t_0,t_1} \, (1+t_1^2)
                          \E(t_0/\sqrt{1+t_1^2} Y - Z)_+^2 \\
                        & =  \min_{t'_0,t_1} \, (1+t_1^2)
                          \E(t'_0 Y - Z)_+^2 
\end{align*}
and the minimum is clearly achieved at $t_1 = 0$. For the second
assertion, a simple calculation reveals that $Y$ and $V$ are
independent and $\mathbb{P}(Y = 1) = 1/2$. By convexity of the mapping
$Y \mapsto (t_0 Y + t_1 V - Z)_+^2$, we have that
\[
  \E \{(t_0 Y + t_1 V - Z)_+^2 \, | \, V, Z\} \ge (\E \{t_0 Y | V, Z\}
  + t_1 V - Z)_+^2 = (t_1 V - Z_+)^2.
\]
Hence, in this case, the miminum in \eqref{eq:gmle} is achieved at
$t_0 = 0$.
\end{proof}

\subsection{Model without intercept}
\label{sec:without}

An analogous result holds for a model without intercept. Its proof is
the same as that of Theorem \ref{thm:main}, only simpler. It is,
therefore, omitted. 
\begin{theorem}
  \label{thm:cor} Assume $\beta_0 = 0$ and consider fitting a model
  without an intercept. If $V$ has the marginal distribution from
  Theorem \ref{thm:main} and is independent from $Z \sim \dnorm(0,1)$,
  then the conclusions from Theorem \ref{thm:main} hold with the
  phase-transition curve given in \eqref{eq:gamma0}.  Hence, the
  location of the phase transition is the same whether we fit an
  intercept or not.
\end{theorem}

\subsection{Comparison with empirical results}
\label{sec:empirical}

We compare our asymptotic theoretical predictions with the results of
empirical observations in finite samples. For a given data set, we can
numerically check whether the data is separated by using linear
programming techniques, see Section \ref{sec:geometry}. (In our setup,
it can be shown that \emph{quasi-complete separation} occurs with zero
probability). To detect separability, we study whether the program
\cite{konis2007linear}
\begin{equation}
  \label{eq:LP}
  \begin{array}{ll}
    \text{maximize} & \quad \sum_{i = 1}^n y_i(b_0 + \bx_i'\bb)\\
    \text{subject to} & \quad y_i (b_0 + \bx_i'\bb)\ge 0, \, i = 1, \ldots, n\\
                    & \quad -1 \le b_0 \le 1, -\bone \le \bb \le \bone
\end{array}
\end{equation}
has a solution or not. For any triplet $(\kappa,\beta_0, \gamma_0)$,
we can thus estimate the probability
$\hat{\pi}(\kappa,\beta_0,\gamma_0)$ that complete separation does not
occur (the MLE exists) by repeatedly simulating data with these
parameters and solving \eqref{eq:LP}.

Below, each simulated data set follows a logistic model with
$n = 4,000$, $p = \kappa \, n$, i.i.d.~Gaussian covariates with
identity covariance matrix (note that our results do not depend on the
covariance $\bSigma$) and $\bbeta$ selected appropriately so that
$\Var(\bx_i'\bbeta) = \gamma_0^2$. We consider a fixed rectangular
grid of values for the pair $(\kappa,\gamma)$ where the $\kappa$ are
equispaced between $0$ and $0.6$ and the $\gamma$'s---recall that
$\gamma = \sqrt{\beta_0^2 + \gamma_0^2}$---are equispaced between $0$
and $10$.  For each triplet $(\kappa,\beta_0, \gamma_0)$, we estimate
the chance that complete separation does not occur (the MLE exists) by
averaging over $50$ i.i.d.~replicates.

Figure \ref{fig:empirical} (a) shows empirical findings for a model
without intercept; that is, $\beta_0 = 0$, and the other regression
coefficients are here selected to have equal magnitude. 
Observe that the MLE existence probability undergoes a sharp phase
transition, as predicted.  The phase transition curve predicted from
our theory (red) is in excellent agreement with the boundary between
high and low probability regions.  Figure \ref{fig:empirical} (b)
shows another phase transition in the setting where $\gamma_0=0$ so
that $\beta_0 = \gamma$. The $y$-axis is here chosen to be the
marginal distribution of the response,
i.e.~$\prob(y_i = 1) = e^\gamma/(1+e^\gamma)$.  Once again, we observe
the sharp phase transition, as promised, and an impeccable alignment
of the theoretical and empirical phase transition curves.  We also see
that when the response distribution becomes increasingly asymmetric,
the maximum dimensionality $\kappa$ decreases, as expected. If $y_i$
has a symmetric distribution, we empirically found that the
MLE existed for all values of $\kappa$ below $0.5$ in all
replications. For $\prob(y_i=1)=0.9$, however, the MLE existed
(resp.~did not exist) if $\kappa < 0.24$ (resp.~if $\kappa > 0.28$) in
all replications. For information, the theoretical value of the phase
transition boundary at $\prob(y_i=1)=0.9$ is equal to
$\kappa = 0.255$.
  
\begin{figure}
\begin{center}
	\begin{tabular}{ccc}
          \includegraphics[width=0.45\textwidth,keepaspectratio]{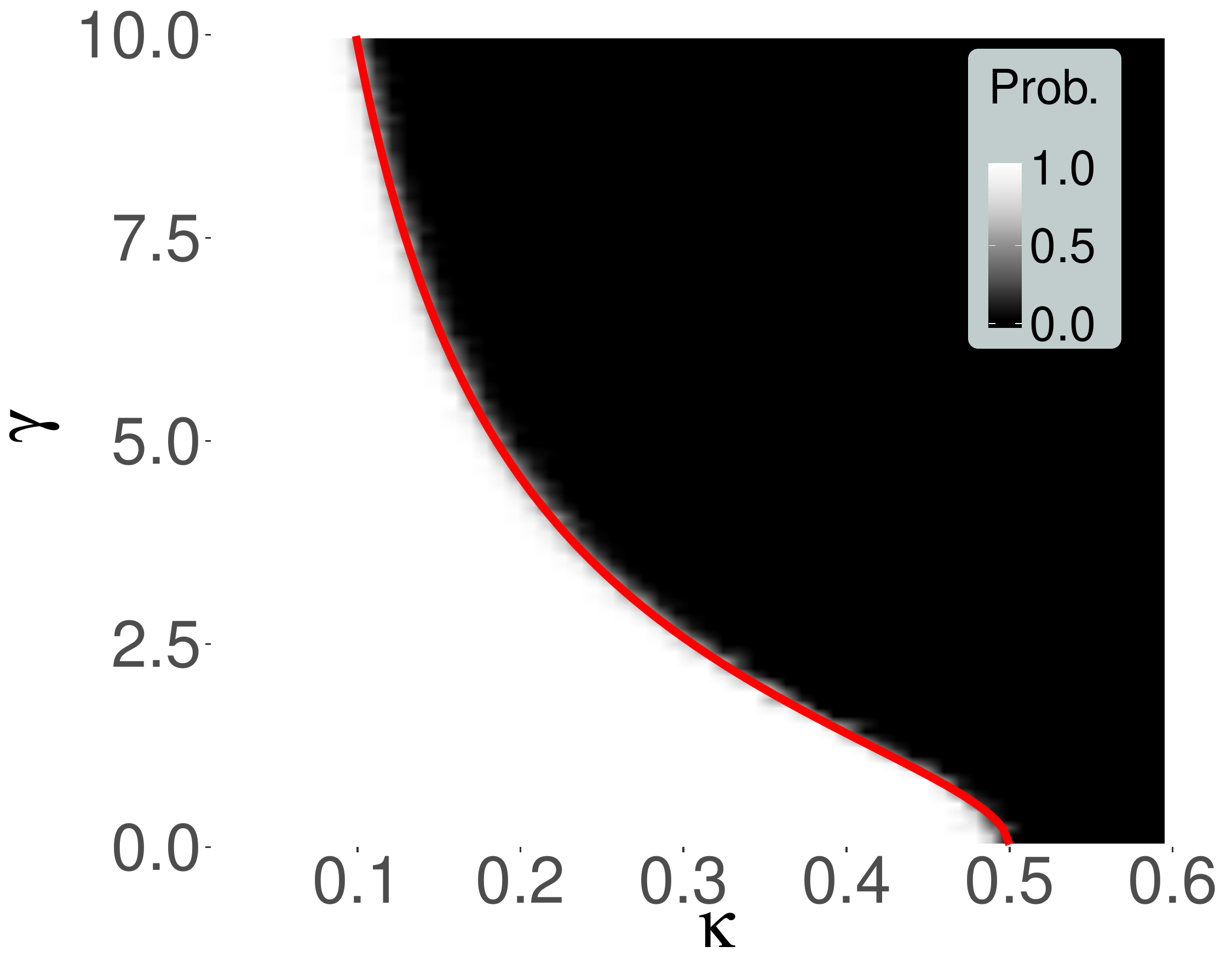} 
          & & \includegraphics[width=0.45\textwidth,keepaspectratio]{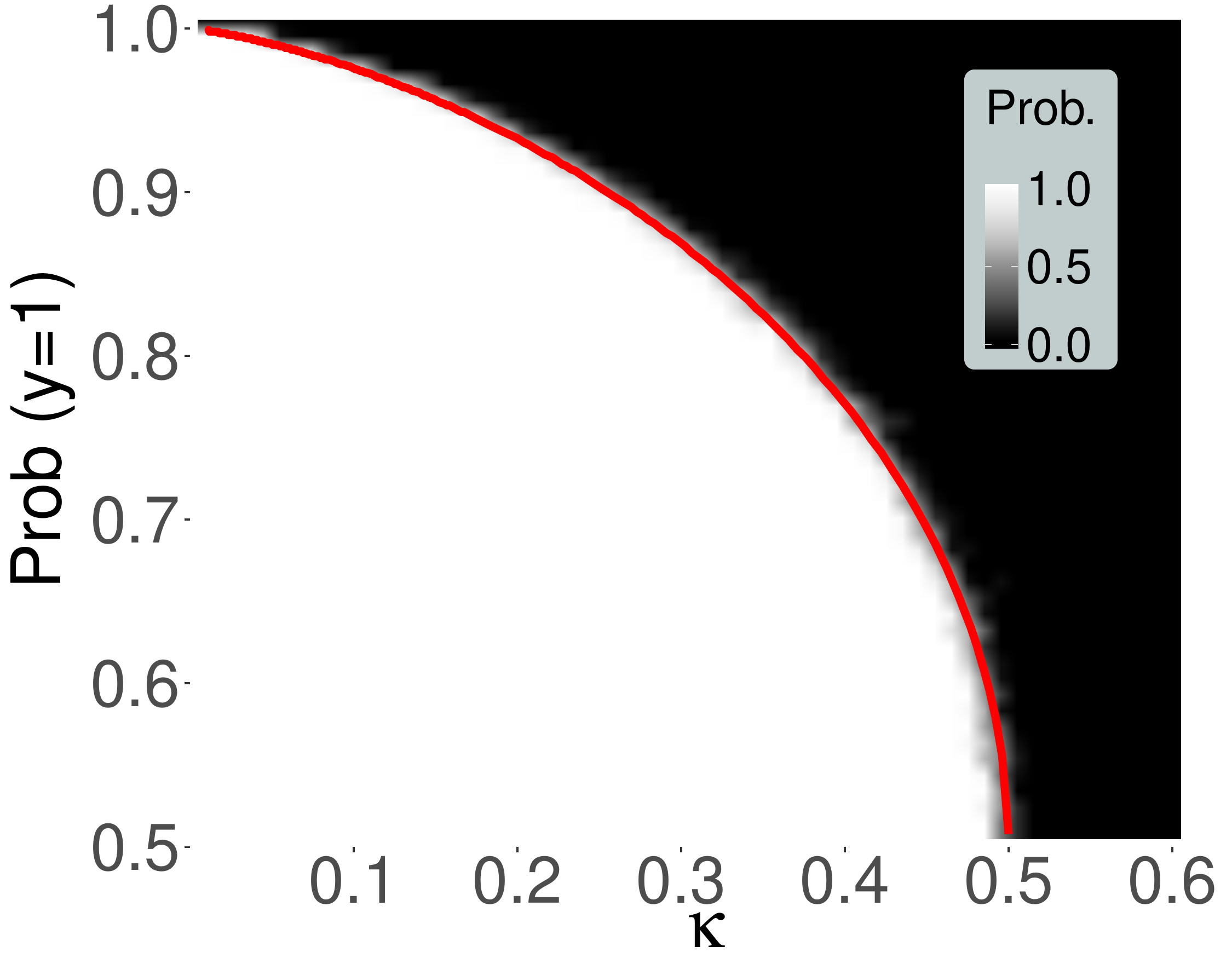} \tabularnewline
              (a) & & (b)   \tabularnewline
	\end{tabular}
\end{center}
\caption{Empirical probability that the MLE exists (black is zero, and
  white is one) estimated from $50$ independent trials for each
  `pixel'. (a) Model without intercept in which $\beta_0 = 0$ and
  $\gamma_0 =\gamma$, with theoretical phase transition curve from
  \eqref{eq:gamma0} in red (this is the same curve as in Figure
  \ref{fig:formula}(a)). (b) Model with $\gamma_0=0$,
  $\beta_0 =\gamma$ and the theoretical phase transition curve from
  \eqref{eq:beta0} in red.  The y-axis is here chosen to be the
  marginal probability $\prob(y_i = 1) = e^\gamma/(1+e^\gamma)$.}
\label{fig:empirical}
\end{figure}

\section{Conic Geometry}
\label{sec:conic_geom}

\renewcommand{\P}{\mathbb{P}}

This section introduces ideas from conic geometry and proves our main
result. 
We shall use the characterization from Albert and Anderson
\cite{Albert1984MLE} reviewed in Section \ref{sec:geometry}; recall
that the MLE does not exist if and only if there is
$(b_0, \bb) \neq \bzero$ such that $y_i \, (b_0 + \bx_i' \bb) \ge 0$
for all $i = 1, \ldots, n$. In passing, the same conclusion holds for
the probit model and a host of related models.

\subsection{Gaussian covariates}

\newcommand{\bz}{{\bm z}}
\renewcommand{\bV}{{\bm V}}

Write $\bx_i \sim \dnorm(\bzero, \bSigma)$ as
$\bx_i = \bSigma^{1/2} \bz_i$, where
$\bz_i \sim \dnorm(\bzero, {\bm I})$. Doing this, we immediately see
that the MLE does not exist if and only if there is $(b_0, \bb) \neq \bzero$
such that
\[
y_i \, (b_0 + \bz_i' \bSigma^{1/2} \bb) \ge 0, \quad \forall i. 
\]
This is equivalent to the existence of
$(b_0, {\bm \theta}) \neq \bzero$ such that
$y_i \, (b_0 + \bz_i' {\bm \theta}) \ge 0$ for all $i$. In words,
multiplication by a non-singular matrix preserves the existence of a
separating hyperplane; that is to say, there is a hyperplane in the
`$z$ coordinate' system (where the variables have identity covariance)
if and only if there is a separating hyperplane in the `$x$
coordinate' system (where the variables have general non-singular
covariance).  Therefore, it suffices to assume that the covariance is
the identity matrix, which we do from now on.

We thus find ourselves in a setting where the $p$ predictors are
independent standard normal variables and the regression sequence is
fixed so that
$\operatorname{Var}(\bx' \bbeta) = \|\bbeta\|^2 = \gamma_0^2$ (the
theorem assumes that this holds in the limit but this does not
matter).  By rotational invariance, we can assume without loss of
generality that all the signal is in the first coordinate; that is,
\[
  \P(y_i = 1 | \bx_i) = \sigma(\beta_0 + \gamma_0 x_{i1})
\]
since this leaves invariant the joint distribution of $(\bx_i, y_i)$.

At this point, it is useful to introduce some notation. Let
$(X_1, \ldots, X_p)$ be independent standard normals. Then
\[
  (\bx_i, y_i) \eqd (X_1, \ldots, X_p; Y),  
\]
where $\P(Y = 1 | X_1, \ldots, X_p) = \sigma(\beta_0 + \gamma_0
X_1)$. It thus follows that
\begin{equation}
  \label{eq:joint}
  (y_i, y_i \, \bx_i) \eqd (Y, V, X_2, \ldots, X_p), \qquad \begin{array}{l} Y, V \sim F_{\beta_0, \gamma_0},\\ (X_2, \ldots, X_p) \sim \dnorm(\bzero, {\bm I_{p-1}}),
                                                 \\ (Y,V) \independent (X_2, \ldots, X_p).\end{array}
\end{equation}
This yields a useful characterization: 
\begin{proposition}
  \label{prop:mle} 
  Let the $n$-dimensional vectors $(\bY, \bV, \bX_2, \ldots, \bX_p)$
  be $n$ i.i.d.~copies of $(Y, V, X_2,$ $\ldots, X_p)$ distributed as
  in \eqref{eq:joint}. Then if $p < n - 1$,
  \begin{equation}
  \label{eq:noMLE}
  \P\{\text{\em no MLE}\} =
  \P\{ \operatorname{span}(\bY, \bV, \bX_2, \ldots, \bX_p) \cap \R_+^n \neq \{\bzero\}\}. 
\end{equation}
Here and below, $\R^n_+$ is the nonnegative orthant.
\end{proposition}
\begin{proof}
  We have seen that there is no MLE if there exists
  $(b_0, b_1, \ldots, b_p) \neq \bzero$ such that
\begin{equation}
  \label{eq:noMLEa}
  b_0 \bY + b_1 {\bm V} + b_2 \bX_2 + \ldots + b_p \bX_p  \ge \bzero.
\end{equation}
By \eqref{eq:joint}, this says that the chance there is no MLE is the
chance of the event \eqref{eq:noMLEa}.  Under our assumptions, the
probability that the $(p-1)$ dimensional subspace spanned by
$\bX_2, \ldots, \bX_p$ non-trivially intersects a fixed subspace of
dimension 2 is zero. Since $(\bY, {\bm V})$ and
$(\bX_2, \ldots, \bX_p)$ are independent, this means that we have
equality in \eqref{eq:noMLEa} with probability zero.
\end{proof}

\subsection{Convex cones}

We are interested in rewriting \eqref{eq:noMLE} in a slightly
different form.  For a fixed subspace ${\cal W} \subset \R^n$,
introduce the convex cone
\begin{equation}
  \label{eq:cone}
  \mathcal{C}({\cal W}) = \{ \bw + \bu: \bw \in {\cal W}, \bu \ge \bzero\}. 
\end{equation}
This is a polyhedral cone, which shall play a crucial role in our
analysis. As we will see, the MLE does not exist if 
$\lspan(\bX_2, \ldots, \bX_p)$ intersects the cone
$\mathcal{C}(\lspan(\bY, \bV))$ in a non-trivial way.


\begin{proposition}
  \label{prop:mletwo} Set ${\cal L} = \lspan(\bX_2, \ldots, \bX_p)$
  and ${\cal W} = \lspan(\bY, \bV)$. Let $\{\text{\em No MLE Single}\}$
  be the event that we can either completely or quasi-separate the
  data points by using the intercept and the first variable only:
  i.e.~${\cal W} \cap \R^n_+ \neq \{\bzero\}$. We have
   \begin{equation}
  \label{eq:noMLEtwo}
  \P\{\text{\em no MLE}\} = \P\{{\cal L} \cap \mathcal{C}({\cal W})
  \neq \{\bzero\} \,\, \text{\em and} \,\, \{\text{\em No MLE
    Single}\}^c\} + \P\{\text{\em No MLE Single}\}.
\end{equation}
An immediate consequence is this:
  \begin{equation}
  \label{eq:noMLEthree}
 0 \le \P\{\text{\em no MLE}\} - \P\{{\cal L} \cap \mathcal{C}({\cal W}) \neq \{\bzero\}\} \le \P\{\text{\em No MLE Single}\}.
\end{equation}
\end{proposition}
\begin{proof}
  If $\{\text{No MLE Single}\}$ occurs, the data is separable and
  there is no MLE. Assume, therefore, that $\{\text{No MLE Single}\}$
  does not occur. We know from Proposition \ref{prop:mle} that we do
  not have an MLE if and only if we can find a nonzero vector
  $(b_0, b_1, \ldots b_p)$ such that
  \[
    b_0 \bY + b_1 \bV + b_2 \bX_2 + \ldots + b_p \bX_p = \bu, \quad \bu
    \ge \bzero, \, \bu \neq \bzero. 
  \]
  By assumption, $b_0 \bY + b_1\bV = \bu$ cannot hold. Therefore,
  $ b_2 \bX_2 + \ldots + b_p \bX_p$ is a non-zero element of
  $\mathcal{C}({\cal W})$. This gives \eqref{eq:noMLEtwo} from which
  \eqref{eq:noMLEthree} easily follows.
\end{proof}

We have thus reduced matters to checking whether ${\cal L}$ intersects
${\cal C}({\cal W})$ in a non-trivial way. This is because we know
that under our model assumptions, the chance that we can separate the
data via a univariate model is exponentially decaying in $n$; that is,
the chance that there is $(b_0, b_1) \neq 0$ such that
$y_i(b_0 + b_1 x_{i1}) \ge 0$ for all $i$ is exponentially small. We
state this formally below. 
\begin{lemma}
 \label{lem:classical} 
 In the setting of Theorem \ref{thm:main}, the event
 $\{\text{\em No MLE Single}\}$ occurs with exponentially small
 probability.
\end{lemma}
\begin{proof}
  We only sketch the argument. We are in a univariate model with
  $\P(y_i = 1 | x_i) = \sigma(\beta_0 + \gamma_0 x_i)$ and $x_i$
  i.i.d.~$\mathcal{N}(0,1)$.  Fix $t_0 \in \R$.  Then it is easy to
  see that the chance that $t_0$ separates the $x_i$'s is
  exponentially small in $n$. However, when the complement occurs, the
  data points overlap and no separation is possible. 
\end{proof}

\subsection{Proof of Theorem \ref{thm:main}}
\label{sec:proofthm}

\newcommand{\bg}{\bm{g}}

To prove our main result, we need to understand when a random subspace
${\cal L}$ with uniform orientation intersects
${\cal C}(\lspan(\bY,\bV))$ in a nontrivial way.  For a fixed subspace
${\cal W} \subset \R^n$, the approximate kinematic formula
\cite[Theorem I]{amelunxen2014living} from the literature on convex
geometry tells us that for any $\epsilon \in (0,1)$
\begin{equation}
  \label{eq:kinematic}
  \begin{array}{lll} 
    p - 1+  \delta(\mathcal{C}({\cal W})) > n + a_\epsilon \sqrt{n} & \Longrightarrow & 
                                                                            \P\{{\cal L} \cap \mathcal{C}({\cal W}) \neq \{\bzero\}\} \ge 1 - \epsilon\\
    p - 1 +  \delta(\mathcal{C}({\cal W})) < n - a_\epsilon \sqrt{n} & \Longrightarrow &
    \P\{{\cal L} \cap \mathcal{C}({\cal W}) \neq \{\bzero\}\} \le \epsilon.
  \end{array}
\end{equation}
We can take $a_\epsilon = \sqrt{8\log(4/\epsilon)}$.
Above, $\delta(\mathcal{C})$ is the {\em statistical dimension} of a
convex cone $\mathcal{C}$ defined as
\begin{equation}
  \label{eq:statistical-dim}
  \delta(\mathcal{C}) := \E \|\Pi_{\mathcal{C}} (\bZ)\|^2 = n - \E \|\bZ - \Pi_{\mathcal{C}} (\bZ) \|^2, \quad \bZ \sim \dnorm(\bzero, {\bm I}_n),
\end{equation}
where $\Pi_{\mathcal{C}}$ is the projection onto $\mathcal{C}$.

We develop a formula for the statistical dimension of the cone
$\mathcal{C}({\cal W})$ of interest to us. 
\begin{lemma}
  \label{lem:statdim}
  Fix ${\cal W} \subset \R^n$. Then with $\bZ$ distributed as in
  \eqref{eq:statistical-dim},
  \begin{equation}
    \label{eq:statdim}
    \delta(\mathcal{C}({\cal W})) =  n - \E \left\{ \min_{\bw \in {\cal W}} \, \|(\bw - \bZ)_+\|^2\right\}.
  \end{equation}
\end{lemma}
\begin{proof}
By definition,
$\delta(\mathcal{C}({\cal W})) = n - \E \operatorname{dist}^2(\bZ,
\mathcal{C}({\cal W}))$, where for a fixed $\bz \in \R^n$,
$\text{dist}^2(\bz, \mathcal{C}({\cal W}))$ is the optimal value of
the quadratic program
\begin{equation*}
  \label{eq:P}
  \begin{array}{ll}
    \text{minimize} & \quad \|\bz - \bw - \bu\|^2\\
    \text{subject to} & \quad \bw \in {\cal W}\\
    & \quad \bu \ge \bzero.
  \end{array}
\end{equation*}
For any $\bw \in {\cal W}$, the optimal value of $\bu$ is given by
$(\bz - \bw)_+$. Hence, the optimal value of the program is
\[
  \operatorname{min}_{\bw \in {\cal W}} \,\, \|\bz - \bw - (\bz -
  \bw)_+\|^2 =  \operatorname{min}_{\bw \in {\cal W}}  \|(\bw - \bz)_+\|^2.
\]
\end{proof}

We claim that this lemma combined with the theorem below establish
Theorem \ref{thm:main}.
\begin{theorem}
  \label{thm:minmin}
  Let $(\bY, \bV)$ be $n$ i.i.d.~samples from 
  $F_{\beta_0,\gamma_0}$.  The random variable
\begin{equation*}
  \label{eq:Q}
  Q_n := \min_{t_0, t_1 \in \R} \, \frac{1}{n} \|(t_0 \bY + t_1 \bV - \bZ)_+\|^2
\end{equation*}
 obeys 
\begin{equation}
  \label{eq:minmin}
  Q_n\, \convP \, h_{\text{\em MLE}}(\beta_0,\gamma_0) = \min_{t_0,t_1} \, \left\{ \E  \, (t_0 Y + t_1 V - Z)_+^2\right\}. 
\end{equation}
In fact, we establish the stronger
statement $Q_n = h_{\text{\em MLE}}(\beta_0,\gamma_0) + O_P(n^{-1/2})$.
\end{theorem}


\newcommand{\filt}{{\cal F}}

Below, we let $\filt$ be the $\sigma$-algebra generated by $\bY$ and
$\bV$.  Set $\epsilon_n = n^{-\alpha}$ for some positive $\alpha$,
$a_n = \sqrt{8\alpha\log(4n)}$, and
define the events
\[
A_n = \{p/n > \E\{Q_n | \filt\} + a_n n^{-1/2} \}, \quad E_n = \{{\cal L} \cap \mathcal{C}({\cal W}) \neq \{\bzero\}\}.
\]

We first show that if $\kappa > h_{\text{MLE}}(\beta_0,\gamma_0)$,
then $\P\{\text{no MLE}\} \goto 1$ or, equivalently,
$\P\{E_n\} \goto 1$.  Our geometric arguments \eqref{eq:kinematic}
tell us that if $A_n$ occurs, then
$\P\{E_n \, | \, \filt\} \ge 1 - \epsilon_n$. This means that
\[
\One{A_n} \le \One{\P\{E_n \, | \, \filt\} \ge 1-\epsilon_n} \le  \P\{E_n \, | \, \filt\} + \epsilon_n. 
\]
Taking expectation gives 
\[
\P\{E_n\} \ge \P\{A_n\} - \epsilon_n. 
\]
Next we claim that 
\begin{equation}
  \label{eq:claim}
  \E\{Q_n | \filt\} \, \convP \, h_{\text{MLE}}(\beta_0,\gamma_0). 
\end{equation}
This concludes the proof since
\eqref{eq:claim} implies that
$\P\{A_n\} \goto 1$ and, therefore, $\P\{E_n\} \goto 1$.  The argument
showing that if $\kappa < h_{\text{MLE}}(\beta_0,\gamma_0)$, then
$\P\{\text{no MLE}\} \goto 0$ is entirely similar and omitted.

It remains to justify \eqref{eq:claim}. Put
$h = \hmle(\beta_0,\gamma_0)$ for short (this is a non-random
quantity), and note that $Q_n - h$ is uniformly integrable (this is
because $Q_n$ is the minimum of an average of $n$
i.i.d.~sub-exponential variables). Hence, if $Q_n$ converges in
probability, it also converges in mean in the sense that
$\E |Q_n - h| \goto 0$. Since
\[
\left| \E \{Q_n | \filt\} - h \right| \le \E \{|Q_n-h| \, | \, \filt\}, 
\]
we see that taking expectation on both sides yields that
$\E \{Q_n | \filt\}$ converges to $h$ in mean and, therefore, in
probability (since convergence in means implies convergence in
probability).


\newcommand{\blambda}{\bm{\lambda}}


\section{Proof of Theorem \ref{thm:minmin}}
\label{sec:proof}

We begin by introducing some notation to streamline our exposition as
much as possible. Define the mapping
$J : \bx \mapsto  \|\bx_+\|^2/2$ and let $\bA$ be the $n \times 2$
matrix with $\by$ and $\bV$ as columns. Next, 
define the random function $F$ and its expectation $f$ as
\[
  F(\blambda) = n^{-1} \, J(\bA \blambda - \bZ), \quad f(\blambda) =
  \E F(\blambda).
\]
Both $F$ and $f$ are convex and it is not hard to see that $f$ is
strictly convex (we will see later that it is, in fact, strongly
convex).  Let $\blambda_\star$ be any minimizer of $F$
($\blambda_\star$ is a random variable) and $\blambda_0$ be the unique
minimizer of $f$ ($\blambda_0$ is not random and finite).  With this
notation, Theorem \ref{thm:minmin} asks us to prove that
\begin{equation}
  \label{eq:claim2}
  F(\blambda_\star) = f(\blambda_0) + O_P(n^{-1/2})  
\end{equation}
and in the rest of this section, we present the simplest argument we
could think of. 

We begin by recording some simple properties of $F$ and $f$. It
follows from $\nabla J(\bx) = \bx_+$ that $\nabla J$ is Lipschitz and
obeys
\[
\|\nabla J(\bx) - \nabla J(\bx_0)\| \le \|\bx - \bx_0\|.
\]
Consequently $F$ is also Lipschitz with constant at most
$n^{-1}\|\bA\|^2 \le n^{-1}(\|\by\|^2 + \|\bV\|^2) = 1 + n^{-1}\|\bV\|^2$.  It is
also a straightforward calculation to see that $f$ is twice
differentiable with Hessian given by
\[
\nabla^2 f(\blambda)  = n^{-1} \, \E \{\bA' {\bm D} \bA\}, \quad {\bm D} = \operatorname{diag}(\One{\bA \blambda - \bZ \ge \bzero}). 
\]
It follows that with $\blambda = (\lambda_0,\lambda_1)$, the Hessian is given by
\begin{equation}
  \nabla^2 f(\blambda) = \begin{bmatrix} \E \{Y^2 \Phi(\lambda_0 Y + \lambda_1 V)\}  
      & \E \{YV \Phi(\lambda_0 Y + \lambda_1 V)\} \\ \E \{YV \Phi(\lambda_0 Y + \lambda_1 
      V)\} & \E\{V^2 \Phi(\lambda_0 Y + \lambda_1 V)\}
  \end{bmatrix}, 
\end{equation}
where $(Y,V)$ is distributed as in Theorem \ref{thm:main} and $\Phi$
is the cdf of a standard normal. We claim that for fixed
$(\beta_0, \gamma_0)$, it holds that 
\begin{equation}
  \label{eq:claimHessian}
 {\alpha_0 {\bm I}_2 \preceq \nabla^2 f(\blambda) \preceq \alpha_1 {\bm I}_2,} 
\end{equation}
uniformly over $\blambda$, where $\alpha_0, \alpha_1$ are fixed
positive numerical constant (that may depend on
$(\beta_0, \gamma_0)$).

Next we claim that for a {\em fixed} $\blambda$, $F(\lambda)$ does not
deviate much from its expectation $f(\blambda)$. This is because
$F(\blambda)$ is an average of sub-exponential variables which are
i.i.d.~copies of $(\lambda_0 Y + \lambda_1 V - Z)_+^2$; classical
bounds \cite[Corollary 5.17]{vershynin2010introduction} give
\begin{equation}
  \label{eq:subexp1}
   \P\{|F(\blambda) - f(\blambda)| \ge t\} 
 \le
  2\exp \left( -c_0 n \, 
    \operatorname{min}\left(\frac{t^2}{c_1^2(1+\|\blambda\|^2)^2},
    \frac{t}{c_1(1+\|\blambda\|^2)}\right)\right),
\end{equation}
where $c_0, c_1$ are numerical constants.
Also, $\nabla F(\lambda)$ does not deviate much from its expectation
$\nabla f(\blambda)$ either because this is also an average of
sub-exponential variables. Hence, we also have
\begin{equation}
  \label{eq:subexp2}
   \P\{\|\nabla F(\blambda) - \nabla f(\blambda)\| \ge t\} 
 \le
  2\exp \left( -c_2 n \, 
    \operatorname{min}\left(\frac{t^2}{c_3^2(1+\|\blambda\|^2)^2},
    \frac{t}{c_3(1+\|\blambda\|^2)}\right)\right),
\end{equation}
where $c_2, c_3$ are numerical constants. In the sequel, we shall make
a repeated use of the inequalities
\eqref{eq:subexp1}--\eqref{eq:subexp2}.

With these preliminaries in place, we can turn to the proof of
\eqref{eq:claim2}.  On the one hand, the convexity of $F$ gives
\begin{equation}
  \label{eq:lower}
F(\blambda_\star) \ge F(\blambda_0) + \langle \nabla F(\blambda_0), \blambda_\star - \blambda_0\rangle. 
\end{equation}
On the other hand, since $\nabla F$ is Lipschitz, we have the upper
bound
\begin{equation}
  \label{eq:upper}
  F(\blambda_\star) \le F(\blambda_0) + \langle \nabla F(\blambda_0), \blambda_\star - \blambda_0\rangle + (1+ \|\bV\|^2/n) \|\blambda_\star - \blambda_0\|^2. 
\end{equation}
Now observe that \eqref{eq:subexp1} gives that 
$$ F(\blambda_0) = f(\blambda_0) + O_P(n^{-1/2}).$$
Also, since $\nabla f(\blambda_0) = \bzero$,  \eqref{eq:subexp2} gives 
\[
\|\nabla F(\blambda_0)\| = O_P(n^{-1/2}). 
\]
Finally, since
$\|\bV\|^2/n \convP \E V^2$, we see from \eqref{eq:lower} and
\eqref{eq:upper} that \eqref{eq:claim2} holds if  $\|\blambda_\star - \blambda_0\| = O_P(n^{-1/4})$.

\begin{lemma}
  \label{lem:sourav}
  We have $\|\blambda_\star - \blambda_0\| = O_P(n^{-1/4})$. 
\end{lemma}
\begin{proof}
  The proof is inspired by an argument in \cite{chatterjee2014new}.
  For any $\blambda \in \R^2$, \eqref{eq:claimHessian} gives
\[
f(\blambda) \ge f(\blambda_0) + \frac{\alpha_0}{2} \|\blambda - \blambda_0\|^2. 
\]
Fix $x \ge 1$. For any $\blambda$ 
on the circle
$C(x) := \{\lambda \in \R^2: \|\lambda - \lambda_0\| = x n^{-1/4}\}$
centered at $\blambda_0$ and of radius $x n^{-1/4}$, we have
\begin{equation}
  \label{eq:useful2}
f(\blambda) \ge f(\blambda_0) + 3y, \qquad y = \frac{\alpha_0 x^2}{6\sqrt{n}}. 
\end{equation}
Fix $z = f(\blambda_0) + y$ and consider the event $E$ defined as
\begin{equation}
  \label{eq:E}
  F(\blambda_0) < z \quad \text{and} \quad \inf_{\lambda \in C(x)} \, F(\blambda) > z. 
\end{equation}
By convexity of $F$, when $E$ occurs, $\blambda_\star$ must lie inside
the circle and, therefore,
$\|\blambda_\star - \blambda_0\| \le x n^{-1/4}$.

It remains to show that $E$ occurs with high probability. Fix $d$
equispaced points $\{\blambda_i\}_{i = 1}^d$ on $C(x)$. Next, take any
point $\blambda$ on the circle and let $\blambda_i$ be its closest
point. By convexity,
\begin{equation}
  \label{eq:convexity2}
  F(\blambda) \ge F(\blambda_i) + \langle \nabla F(\blambda_i), \blambda
  - \blambda_i\rangle \ge F(\blambda_i) - \|\nabla F(\blambda_i)\|
  \|\blambda - \blambda_i\|. 
\end{equation}
On the one hand, $\|\blambda - \blambda_i\| \le \pi x n^{-1/4}/d$.  On
the other, by \eqref{eq:subexp2} we know that if we define $B$ as
\[
  B := \left\{\max_i \|\nabla F(\blambda_i) - \nabla
  f(\blambda_i)\|_2 \ge x n^{-1/2}\right\}
\]
then
\begin{equation}
  \label{eq:Bc}
  \P\{B^c\} \le 
  2d \,\exp \left( -c_2
    \operatorname{min}\left(\frac{x^2}{c_3^2(1+\max_i \|\blambda_i\|^2)^2},
      \frac{\sqrt{n} x}{c_3(1+\max_i \|\blambda_i\|^2)}\right)\right).
\end{equation}
Also, since $\|\nabla^2 f\|$ is bounded \eqref{eq:claimHessian} and
$\nabla f(\blambda_0) = 0$,
\[
  \|\nabla f(\blambda_i)\|_2 \le \alpha_1 \|\blambda_i - \blambda_0\| =
  \alpha_1 \, x n^{-1/4}. 
\]
For $n$ sufficiently large, this gives that on $B$, 
\[
  \|\nabla F(\blambda_i)\| \|\blambda - \blambda_i\| \le C \, {y}/{d}
\]
for some numerical constant $C$. Choose $d \ge C$. Then it follows
from \eqref{eq:convexity2} that on $B$,
\[ \inf_{\blambda \in C(x)} \, F(\blambda) \ge \min_i F(\blambda_i) - y. 
\]
It remains to control the right-hand side above. To this end, observe
that
\[
F(\blambda_i) > f(\blambda_i) - y \quad \implies \quad F(\blambda_i) - y > f(\blambda_0) + y = z
\]
since $f(\blambda_i) \ge f(\blambda_0) + 3y$ by \eqref{eq:useful2}.
Hence, the complement of the event $E$ in \eqref{eq:E} has probability
at most
\[
  \P\{E^c\} \le \P\{B^c\} + \P\{F(\blambda_0) \ge f(\blambda_0) + y\} + \sum_{i =
    1}^d \P\{F(\blambda_i) \le f(\blambda_i) - y\}. 
\]
The application of \eqref{eq:Bc} and that of \eqref{eq:subexp1} to the last 
two terms in the right-hand side concludes the proof. 
\end{proof}

\section{Conclusion}
\label{sec:conclusion}

In this paper, we established the existence of a phase transition for
the existence of the MLE in a high-dimensional logistic model with
Gaussian covariates. We derived a simple expression for the
phase-transition boundary when the model is fitted with or without
an intercept. Our methods use elements of convex geometry, especially the
kinematic formula reviewed in Section \ref{sec:proofthm}, which is a
modern version of Gordon's escape through a mesh theorem
\cite{Gordon88}. It is likely that the phenomena and formulas
derived in this paper hold for more general covariate distributions,
and we leave this to future research.

\small
\subsection*{Acknowledgements}

P.~S.~was partially supported by the Ric Weiland Graduate Fellowship
in the School of Humanities and Sciences, Stanford
University. E.~C.~was partially supported by the Office of Naval
Research under grant N00014-16-1-2712, by the National Science
Foundation via DMS 1712800, by the Math + X Award from the Simons
Foundation and by a generous gift from TwoSigma. E.~C.~would like to
thank Stephen Bates and Nikolaos Ignatiadis for useful comments about
an early version of the paper.

\bibliographystyle{plain}
\bibliography{bibfileLR}

\end{document}